\documentclass[a4paper,10 pt]{ieeeconf}
\usepackage[english]{babel}
\usepackage[bottom=19.1mm,left=19.1mm,right=13.1mm,top=36.7mm]{geometry}
\IEEEoverridecommandlockouts

\usepackage{microtype}
\usepackage{setspace}
\usepackage{blindtext}
\usepackage{siunitx}
\newcommand{\subparagraph}{}
\usepackage{titlesec}
\usepackage{titlecaps}
\Addlcwords{or with if of and to -off off for in versus subject the} 

\titleformat{\section}
{\normalfont\normalsize\centering\scshape}
{\thesection}
{1em}
{\titlecap} 

\titleformat{\subsection}
{\normalfont\normalsize\itshape}
{\thesubsection}
{1em}
{\titlecap} %

\usepackage{cite}

\usepackage{subcaption}
\usepackage{graphicx}
\usepackage{subfloat}
\usepackage{fancyhdr} 
\usepackage{color}
\usepackage{epsfig}
\graphicspath{{Images/}}
\usepackage{tikz}
\usepackage{tikzscale}
\usepackage{scalerel}
\usetikzlibrary{arrows}
\usetikzlibrary{plotmarks}
\usetikzlibrary{svg.path}
\usetikzlibrary{shapes.multipart}
\usepackage{pgfplots}
\pgfplotsset{compat=newest}
\pgfplotsset{every axis/.append style={
		label style={font=\Large},
		tick label style={font=\large}  
}}
\tikzstyle{int}=[draw, fill=black!10, minimum size=5em,thick]
\tikzstyle{init} = [pin edge={to-,thick,black}]

\usepackage{array}
\usepackage{stfloats}

\usepackage{amsmath,amssymb,amsfonts,amsthm}
\usepackage{bbm,dsfont}
\usepackage{relsize}
\usepackage{nicefrac}
\usepackage{bbm}
\usepackage{mathtools}
\usepackage[mathscr]{eucal}

\usepackage[ruled]{algorithm}
\usepackage{algpseudocode}
\usepackage{multirow}
\makeatletter
\gdef\Shortstack{\@ifnextchar[\@Shortstack{\@Shortstack[c]}}
\gdef\@Shortstack[#1]#2{%
	\leavevmode
	\vbox\bgroup
	\baselineskip-\p@\lineskip 3\p@
	\let\mb@l\hss\let\mb@r\hss
	\expandafter\let\csname mb@#1\endcsname\relax
	\let\\\@stackcr\setlength{\baselineskip}{#2}%
	\@ishortstack}
\makeatother
\usepackage{booktabs}
\usepackage{tabularx}
\let\labelindent\relax
\usepackage{enumitem}

\usepackage{float}
\usepackage{textpos}

\makeatletter
\let\NAT@parse\undefined
\makeatother
\usepackage{url}
\usepackage{hyperref}
\usepackage[capitalize]{cleveref}
\newcommand\orcidicon[1]{\href{https://orcid.org/#1}{\includegraphics[scale=0.05]{orcid}}}

\newcommand{\isExtended}[2]{#2}
\newcommand{\myParagraph}[1]{{\bf #1.}}

\newcommand{\Real}[1]{ { {\mathbb R}^{#1} } }
\newcommand{\Realp}[1]{ { {\mathbb R}^{#1}_+ } }

\DeclareMathOperator*{\argmin}{arg\,min}

\DeclarePairedDelimiter\floor{\lfloor}{\rfloor}

\newcommand{\gauss}{\mathcal{N}}

\newcommand{\x}[2]{x_{#1}(#2)}
\newcommand{\xdot}[2]{\dot{x}_{#1}(#2)}
\newcommand{\xdiff}[2]{d{x}_{#1}(#2)}
\newcommand{\xinf}{x_{\infty}}
\newcommand{\xtilde}[2]{\tilde{x}_{#1}(#2)}
\newcommand{\xdifftilde}[2]{d{\tilde{x}}_{#1}(#2)}
\newcommand{\xinftilde}{\tilde{x}_{\infty}}
\newcommand{\xbar}[2]{\bar{x}_{#1}(#2)}
\newcommand{\xbardiff}[2]{d\bar{x}_{#1}(#2)}

\renewcommand{\u}[2]{u_{#1}(#2)}
\newcommand{\ubar}[2]{\bar{u}_{#1}(#2)}
\newcommand{\noise}[2]{w_{#1}(#2)}
\newcommand{\dnoise}[2]{dw_{#1}(#2)}
\newcommand{\noisebar}[2]{\bar{w}_{#1}(#2)}
\newcommand{\noisetilde}[2]{\tilde{w}_{#1}(#2)}
\newcommand{\meas}[2]{y_{#1}(#2)}
\newcommand{\optvar}{\sigma^{2,*}_{\textit{ss}}}
\newcommand{\opteig}{\lambda^*}

\newcommand{\taun}{\tau_n}
\newcommand{\consMatrix}{\dfrac{\mathds{1}_N\mathds{1}_N^{\top}}{N}}
\newcommand{\optvarx}{\tilde{\sigma}^{2,*}_{x,\textit{ss}}}

\theoremstyle{plain}
\newtheorem{thm}{Theorem}

\newtheorem{prop}{Proposition}
\newtheorem{lemma}{Lemma}
\theoremstyle{definition}

\newtheorem{prob}{Problem}
\newtheorem{ass}{Assumption}
\theoremstyle{remark}
\newtheorem{rem}{Remark}

\algdef{SE}[DOWHILE]{Do}{doWhile}{\algorithmicrepeat}[1]{\algorithmicuntil\ #1}%

\renewcommand{\algorithmiccomment}[1]{\bgroup\hfill//~#1\egroup}

\addto\captionsenglish{}
\addto\captionsenglish{}

\newcommand{\linkToPdf}[1]{\href{#1}{\blue{(pdf)}}}
\newcommand{\linkToPpt}[1]{\href{#1}{\blue{(ppt)}}}
\newcommand{\linkToCode}[1]{\href{#1}{\blue{(code)}}}
\newcommand{\linkToWeb}[1]{\href{#1}{\blue{(web)}}}
\newcommand{\linkToVideo}[1]{\href{#1}{\blue{(video)}}}
\newcommand{\linkToMedia}[1]{\href{#1}{\blue{(media)}}}
\newcommand{\award}[1]{\xspace} 
\newcommand{\eg}{\emph{e.g.,}\xspace}
\newcommand{\ie}{\emph{i.e.,}\xspace}
\addto\extrasenglish{}
\addto\extrasenglish{}
\addto\extrasenglish{}

\title{\LARGE \bf\titlecap{optimal network topology of multi-agent systems subject to computation and communication latency}\isExtended{}{ (with proofs)}*}

\author{Luca~Ballotta$ ^1 $, Mihailo R.\ Jovanovi\'c$ ^2 $ and Luca~Schenato$ ^1 $%
	\thanks{*This work has been partially funded by the CARIPARO Foundation Visiting Programme “HiPeR” and by the Italian Ministry of Education, University and Research (MIUR) through the PRIN project no. 2017NS9FEY entitled ``Realtime Control of 5G Wireless Networks: Taming the Complexity of Future Transmission and Computation Challenges'' and through the initiative "Departments of Excellence" (Law 232/2016). The views and opinions expressed in this work are those of the authors and do not necessarily reflect those of the funding institutions.}%
	\thanks{$ ^1 $L. Ballotta and L. Schenato are with the Department of Information Engineering, University of Padova, 35131 Padova, Italy
		{\tt\small \{ballotta, schenato\}@dei.unipd.it}}%
	\thanks{$ ^2 $M. R.\ Jovanovi\'c is with the Ming Hsieh Department of Electrical and Computer Engineering,
		University of Southern California, Los Angeles, CA 90089 USA {\ttfamily\small mihailo@usc.edu}}
}

\begin{document}
	
	\newgeometry{bottom=19.1mm,left=19.1mm,right=13.1mm,top=43mm}
	\maketitle
	\isExtended{}{%
		\begin{textblock}{10}(2,-4.5)
			\footnotesize
			\centering
			\setstretch{1}
			\textcopyright 2021 IEEE.  
			Personal use of this material is permitted.  
			Permission from IEEE must be obtained for all other uses, in any current or future media, including reprinting/republishing this material for advertising or promotional purposes, creating new collective works, for resale or redistribution to servers or lists, or reuse of any copyrighted component of this work in other works.
		\end{textblock}
		\begin{textblock}{10}(2,-3.5)
			\footnotesize
			\centering
			\setstretch{1}
			This paper has been accepted for the 29th Mediterranean Conference on Control and Automation.\\
			Please cite the paper as: L. Ballotta, M. R. Jovanovi\'c, and L. Schenato,\\
			“Optimal Network Topology of Multi-Agent Systems subject to Computation and Communication Latency”,\\
			29th Mediterranean Conference on Control and Automation (MED), 2021.
		\end{textblock}
	}
	\thispagestyle{empty}
	\pagestyle{empty}

\begin{abstract}
	We study minimum-variance feedback-control design for 
	a networked control system with retarded dynamics,
	where inter-agent information exchange is subject to latency.
	We prove that such a control design can be solved 
	efficiently for circular formations
	and compute near-optimal control gains in closed form.
	Also, we show that the centralized control
	is in general a poor design choice when
	adding communication links to the network
	increases the latency,
	and propose a control-driven optimization of the network topology.
	
	
	\isExtended{
	}{\vspace{-.5cm}}
	
\end{abstract}

\section{Introduction}

Large-scale networked control systems
have been scaling up both in terms of number of agents and spatial distribution
over the last years,
taking advantage of new communication protocols for massive systems,
\eg 5G~\cite{li20185g,biral2015challenges},
and of advances in embedded electronics~\cite{suleiman2018navion}
enhancing the computing performance of the network agents~\cite{yi2015survey,shi2016edge}.
Indeed, distributed sensing and computation represent
the true asset of such systems,
whose tasks would otherwise be infeasible.
Since the concept of multi-agent systems was introduced to the literature,
decentralized control techniques have been developed
to adapt classical approaches to network applications~\cite{BAKULE200887}.
A few examples can be found in control of vehicular formations~\cite{chehardoli2018adaptive,delimpaltadakis2018decentralized}
and of traffic congestion in vehicular networks~\cite{7835931},
ground and aerial robot swarms~\cite{yuan2017outdoor,schiano2016rigidity,li2017decentralized}
and UAV-UGV cooperative teams~\cite{arbanas2018decentralized}.

One of the main issues in multi-agent systems
is latency, due to communication constraints,
limited computational power, or delayed sensing and actuation.
In particular, limited bandwidth and finite channel capacity
pose a practical limitation to the scalability of large-scale systems
and when the information exchange among agents involves heavy data,
as with multimedia or in federated learning~\cite{shi2020joint,2020arXiv200601816O}.

Several works deal with decentralized control in the presence of latency.
~\cite{8844785,doi:10.1177/1077546318791025} design control laws with convergence conditions
for platoons under various network topologies.
\cite{7994706,ZONG2019412} study consensus of multi-agent systems with time-delays,
showing its dependence on parameters such as dynamics and communication graph and
giving conditions for mean-square stability.
\cite{ren2017finite} is concerned with finite-time stability of discrete-time neutral-delay systems.

Drawing inspiration from this approach,
we design minimum-variance control of networked systems
where the information exchange among agents is subject to latency,
and we assume that the delays
increase with the number of links.
Such a situation occurs, \eg
when the total bandwidth available to the network is fixed \textit{a priori}~\cite{garcia2016periodic},
so that additional links reduce the bandwidth for all communications,
or when multi-hop transmissions yield cumulative delays~\cite{gupta2009delay}.
Also, we demonstrate that
adding links beyond a certain threshold
worsens the performance,
because the latency increase overtakes the benefit accrued by a larger amount of feedback information.
In particular, centralized control (the complete graph)
is in general a poor choice.
To the best of our knowledge,
this analysis is new in the existing literature.

To approach such control design in the presence of latency,
we consider circular formations.
Such systems are deployed for, \eg
source seeking by smart mobile sensors~\cite{moore2010source,brinon2015distributed}
or target tracking by robots~\cite{arranz2009translation,ma2013cooperative}.
This allows to get analytical results
which give insights on the design,
while extension to more general topologies is straightforward.

\subsection{contribution and paper outline}
The paper is organized as follows.
\autoref{sec:setup} introduces the system model,
detailing results on circular formations (\autoref{sec:eigs-subsystems})
and delay systems (\autoref{sec:stochastic-retarded-diff-eq}).
\autoref{sec:optimization} formalizes the control problem:
namely, we optimize the feedback gains in order to minimize the steady-state scalar variance
of the system.
We first solve the problem with a single parameter (equal gains) in~\autoref{sec:single-param}
and then let multiple parameters in~\autoref{sec:multiple-params}.
In the latter, we prove that an efficient solution
coincides with a particular single-parameter configuration.
In~\autoref{sec:optimal-n}, we discuss the control performance under varying network topologies,
proposing an optimization that takes into account how the delays increase with the number of links.
Specifically, we show that using all communication links (centralized control) yields in general poor performance.
Finally, concluding remarks are drawn in~\autoref{sec:conclusions}, together with proposals for further study.
	\restoregeometry

\section{setup}\label{sec:setup}


\subsection{system model}\label{sec:single-integrator-model}
We consider a circular formation composed of $ N $ agents,
where each agent $ i\in[N] $ is modeled as a single integrator:
\begin{equation}\label{eq:systemDynamicsOrig}
	\begin{aligned}
		\xbardiff{i}{t} = \ubar{i}{t} dt + d\noisebar{i}{t}
	\end{aligned}
\end{equation}
where $ \xbar{i}{t}, \ubar{i}{t}\in\Real{} $ are
the state and the control input of agent $ i $ at time $ t $, respectively,
and $ \noisebar{i}{t} $ is a standard Brownian motion.
Each  agent computes its control input $ \ubar{i}{t} $ 
exploiting the mismatches between its own state (measured) 
and other agents' states (received via wireless). 
\begin{ass}\label{ass:hypothesis}
	Agent $ i $ receives state measurements from the $ n < \nicefrac{N}{2} $ agent pairs 
	located $ \ell $ positions ahead and behind,
	$ \ell = 1,...,n $.
	Measurements are received after the delay $ \taun = f(n)\tau_{\textit{min}} $ from the corresponding state sampling,
	where $ f(n) $ is a non-decreasing sequence
	and $ \tau_{\textit{min}} \in \Realp{} $ is a constant.
\end{ass}
\begin{rem}
	The time $ \taun $ embeds both the communication delay,
	due to channel constraints,
	and the computation delay,
	arising if the agents preprocess the acquired measurements.
	In practice, $ f(n) $ is to be estimated or learned from data.
\end{rem}
The control input at time $ t $ has the following structure:
\begin{subequations}\label{eq:controlInput}
	\begin{equation}\label{eq:controlInputFeedforward}
		\ubar{i}{t} = \u{d}{t} + \u{i}{t}
	\end{equation}
	\begin{equation}\label{eq:controlInputFeedback}
		\u{i}{t} = - \sum_{\ell=1}^{n}k_\ell\left[\meas{i,i-\ell}{t-\taun} + \meas{i,i+\ell}{t-\taun}\right]
	\end{equation}
\end{subequations}
where $ \u{d}{t} $ is the desired reference (feedforward control input), 
$ \u{i}{t} $ is the feedback control input for agent $ i $,
$ k_\ell $ is $ \ell $-th the feedback gain, $ \ell = 1,...,n $, and
\begin{equation}\label{eq:meas}
	\meas{i,i\pm\ell}{t} = 
		\begin{cases}
			\xbar{i}{t} - \xbar{i\pm\ell}{t}, & 0<i\pm\ell\le N\\
			\xbar{i}{t} - \xbar{i\pm\ell\mp N}{t}, & \mbox{otherwise}
		\end{cases}
\end{equation}
are the state mismatches computed with the received measurements,
which in~\eqref{eq:controlInputFeedback} are delayed according to~\cref{ass:hypothesis}.\\
System~\eqref{eq:systemDynamicsOrig}--\eqref{eq:controlInput} can be modeled in vector-matrix form as
\begin{equation}\label{eq:stateSpaceSystem}
	d\xbar{}{t} = \mathds{1}_N\u{d}{t} dt - K\xbar{}{t-\taun}dt + d\noisebar{}{t}
\end{equation}
where $ \mathds{1}_N $ is the vector of all ones of dimension $ N $,
$ \xbar{}{t} $ contains the stacked states of all agents at time $ t $,
and $ d\noisebar{}{t} \sim \mathcal{N}(0,I_Ndt) $. 
Denoting the gain vector as $ k^\top = [k_1,\dots,k_n] $,
the circulant feedback matrix $ K \in\Real{N\times N} $ can be written as
\begin{subequations}\label{eq:feedbackMatrixMulti}
	\begin{equation}
		K = K_f(k) + K_f^\top(k)
	\end{equation}
	\begin{equation}\label{eq:feedbackMatrixGains}
		K_f(k) = 
		\mathrm{circ}\left(\sum_{\ell=1}^nk_\ell, -k_1, \dots, -k_n, 0,  \dots, 0\right)
	\end{equation}
\end{subequations}
where $ \mathrm{circ}(a_1,...,a_N) $ is the circulant matrix with the vector $ [a_1,...,a_N] $ as first row.
For example, if $ N = 6 $ and $ n = 2 $, 
\begin{equation}\label{eq:exampleKMultiParam}
	K = 
	\begin{bmatrix}
		\ast & \circ & \bullet & 0 & \bullet & \circ \\
		\circ & \ast & \circ & \bullet & 0 & \bullet \\
		\bullet & \circ & \ast& \circ & \bullet & 0 \\
		0 & \bullet & \circ & \ast & \circ & \bullet \\
		\bullet & 0 & \bullet & \circ & \ast & \circ \\
		\circ & \bullet & 0 & \bullet & \circ & \ast
	\end{bmatrix}
\end{equation}




\subsection{decoupling the error dynamics}\label{sec:eigs-subsystems}
Let us decompose the state as $ \xbar{}{t} = \xbar{m}{t} + \x{}{t} $,
where all components of $ \xbar{m}{t} $ equal the mean of $ \xbar{}{t} $
and $ \x{}{t} $ represents the current mismatches of the agents' states:
\begin{gather}\label{eq:stateDecouplingMeanErrors}
\xbar{m}{t} = P\xbar{}{t}, \qquad \x{}{t} = \Omega\xbar{}{t} \\
P \doteq \consMatrix, \qquad \Omega \doteq I-P
\end{gather}
The mean-vector dynamics equation reads (c.f.~\eqref{eq:stateSpaceSystem})
\begin{equation}\label{eq:stateMeanDynamics}
\xbardiff{m}{t} = P\xbardiff{}{t} = \mathds{1}_N\u{d}{t}dt + P d\noisebar{}{t}
\end{equation}
because the columns of $ K $ belong to the kernel of $ P $.
The error dynamics is
\begin{align}\label{eq:stateErrorDynamics}
\begin{split}
\xdiff{}{t} = \Omega\xbardiff{}{t}  &= -\Omega K\xbar{}{t-\taun}dt + \Omega d\noisebar{}{t} \\
&= -K\Omega\xbar{}{t-\taun}dt + d\noise{}{t} \\
&= -K\x{}{t-\taun}dt + d\noise{}{t}
\end{split}
\end{align}
where $ \noise{}{t} \!=\! \Omega\noisebar{}{t} $. 
From now on, we focus on system~\eqref{eq:stateErrorDynamics}. 

We exploit symmetry to diagonalize $ K = T\Lambda T^\top $,
with $ T $ orthogonal and $ i $-th eigenvalue $ \lambda_i = \Lambda_{ii} $ equal to (c.f.~\cite{circulant})
\begin{equation}\label{eq:eigenvaluesCirculant}
	\lambda_i = 2\sum_{\ell=1}^nk_\ell\left(1-\cos\left(\dfrac{2\pi (i-1) \ell}{N}\right)\right), \ i = 1,...,N
\end{equation}
The dynamics of the new state $ \xtilde{}{t} = T^{\top}\x{}{t} $ reads
\begin{equation}\label{eq:scalar-subs-eigs}
	\xdifftilde{i}{t} = -\lambda_i\xtilde{i}{t-\taun}dt + d\noisetilde{i}{t}
\end{equation}
where $ d\noisetilde{}{t} \sim \gauss(0,Qdt) $ with covariance matrix
\begin{equation}\label{eq:noiseCovariance}
	Q = T^{\top}\Omega T = (\Omega T)^{\top} (\Omega T) = \left[
	\begin{array}{c|c}
		0 & 0 \\
		\hline
		0 & I_{N-1}
	\end{array}
	\right]
\end{equation}
because the leftmost eigenvector in $ T $, associated with $ \lambda_1 = 0 $, 
belongs to the kernel of $\Omega$,
and all the other eigenvectors in $ T $ are also in $ \Omega $.
In particular, the coordinate $ \xtilde{1}{t} $ (\ie the mean of $ x(t) $),
has trivial dynamics
and does not affect system~\eqref{eq:stateErrorDynamics}, 
whose state coordinates converge to random variables with expectation equal to $ \xtilde{1}{0} $.
For the sake of simplicity,
in the following we assume that $ x(0) $ is zero mean.

\subsection{minimum-variance control of scalar delay systems}\label{sec:stochastic-retarded-diff-eq}
Let us consider the stochastic retarded differential equation
\begin{equation}\label{eq:retarded-diff-eq}
	\begin{aligned}
		\xdiff{}{t} &= -a\x{}{t-\tau}dt + \dnoise{}{t},  \quad t \ge 0\\
		\x{}{t} &= \x{0}{t}, \quad t\in [-\tau,0]
	\end{aligned}
\end{equation}
where $ \x{}{t},a \in\Real{} $,
$ \tau\in\Realp{} $ is the delay,
$ \noise{}{t} $ is standard Brownian noise,
and $ \x{0}{t} $, $ t\in[-\tau,0] $, is the initial condition.
\begin{thm}[\citen{KuchlerLangevinEqs}]\label{thm:retarded-eq-steady-state}
	Eq.~\eqref{eq:retarded-diff-eq} admits a steady-state solution $ \x{\textit{ss}}{t} $ if and only if 
	\begin{equation}\label{eq:retarded-eq-steady-state-condition}
		a \in \left(0,\dfrac{\pi}{2\tau}\right)
	\end{equation}
	If it exists, $ \x{\textit{ss}}{t} $ is unique and it is a zero-mean Gaussian process with variance
	\begin{equation}\label{eq:steady-state-variance}
		\sigma^2_{\textit{ss}} = \int_{0}^{\infty} x_d^2(s)ds = \dfrac{1+\sin(a\tau)}{2a\cos(a\tau)}
	\end{equation}
	where $ \x{d}{t} $ is the so-called \emph{fundamental solution} of the deterministic equation corresponding to~\eqref{eq:retarded-diff-eq}:
	\begin{equation}\label{eq:retarded-deterministic-diff-eq}
		\begin{aligned}
			\xdot{}{t} &= -a\x{}{t-\tau}, \quad t \ge 0\\
			\x{}{t} &= \delta(t), \quad t\in [-\tau,0]
		\end{aligned}
	\end{equation}
	If~\eqref{eq:retarded-eq-steady-state-condition} holds, $ \x{d}{t} $ and $ \dot{x}_d(t) $ are exponentially stable.
\end{thm}
\begin{lemma}\label{lemma:convexVariance}
	The variance $ \sigma^2_{\textit{ss}} $
	in~\eqref{eq:steady-state-variance}
	is strictly convex in $ a $.
\end{lemma}
\begin{proof}
	Tedious but standard computations yield $ (\sigma^2_{\textit{ss}})'' > 0 $.
	See 
	\isExtended{%
		Appendix A in the preprint~\cite{2021arXiv210110394B}
		}{%
		\cref{app:convexVariance}
	}%
	for details.
\end{proof}
Let us consider the scalar system with delayed dynamics
\begin{equation}\label{eq:systemRetardedEquation}
	d\x{}{t} = \u{}{t}dt + d\noise{}{t}
\end{equation}
and control input $ \u{}{t} = -\lambda\x{}{t-\tau} $.
We wish to design the gain $ \lambda $ in order to minimize the steady-state variance of $ \x{}{t} $.
If $ \tau = 0 $,
the (infeasible) optimum is $ +\infty $ and $ \lambda $ can be increased arbitrarily.
Instead, if $ \tau>0 $, the following holds.
\begin{figure}
	\centering
	\begin{minipage}[l]{.48\linewidth}
		\centering
		\includegraphics[height=.83\linewidth]{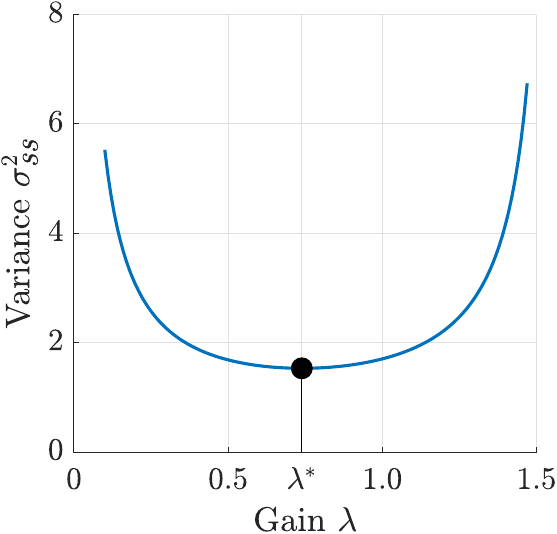}
		\caption{Variance $ \sigma^2_{\textit{ss}} $ as a function of the gain $\lambda$ ($ \tau = 1 $).}
		\label{fig:sigma_of_lambda}
	\end{minipage}%
	\hfill
	\begin{minipage}[r]{.48\linewidth}
		\centering
		\includegraphics[height=.83\linewidth]{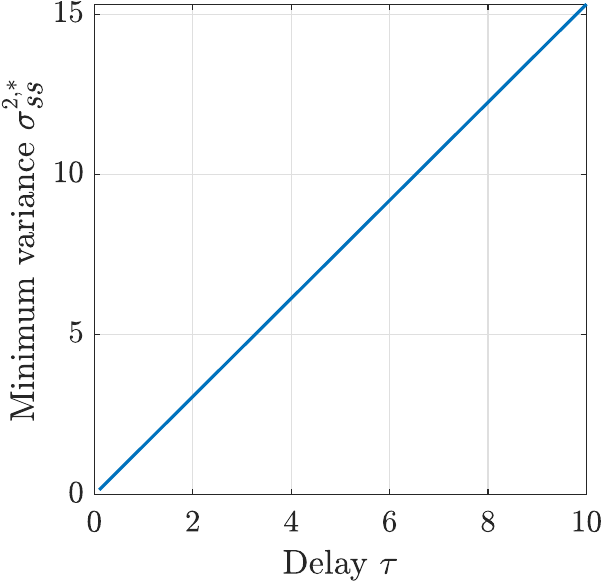}
		\caption{Minimum variance $ \sigma^{2,*}_{\textit{ss}} $ as a function of $\tau$.}
		\label{fig:sigma_star_of_tau}
	\end{minipage}\\
	\vspace{2mm}
	\begin{minipage}{\linewidth}
		\centering
		\includegraphics[height=.41\linewidth]{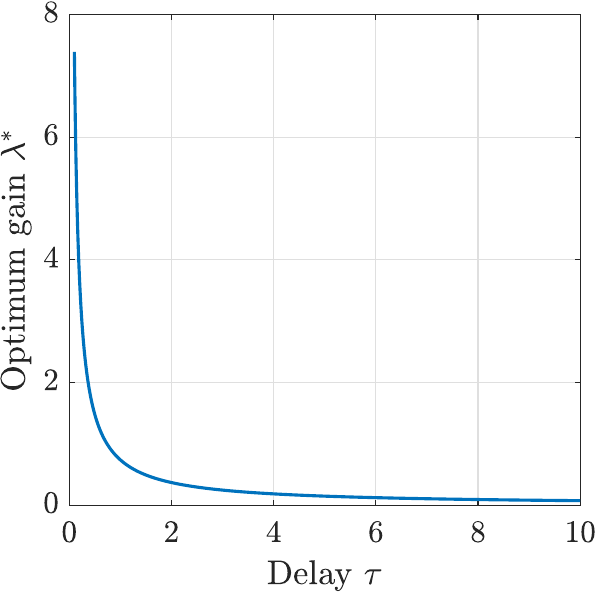}
		\caption{Optimal gain $ \lambda^{*} $ as a function of the delay $\tau$.}
		\label{fig:lambda_star_of_tau}
	\end{minipage}
\end{figure}
\begin{prop}\label{prop:optimumGain}
	Let system~\eqref{eq:systemRetardedEquation} with input $ \u{}{t} = -\lambda\x{}{t-\tau} $, $ \tau > 0 $,
	and steady-state variance of $ x(t) $ as per~\eqref{eq:steady-state-variance}:
	\begin{equation}\label{eq:variance}
		\sigma^2_{\textit{ss}}(\lambda) = \dfrac{1+\sin(\lambda \tau)}{2\lambda\cos(\lambda \tau)}
	\end{equation}
	Then, the minimum-variance control problem
	\begin{equation}\label{eq:opt-var-eig}
		\argmin_{\lambda\in\left(0,\frac{\pi}{2\tau}\right)}\sigma^2_{\textit{ss}}(\lambda)
	\end{equation}
	has the following unique solution:
	\begin{equation}\label{eq:optGainExplicit}
		\opteig = \dfrac{\beta^*}{\tau}, \qquad \optvar \doteq \sigma^2_{\textit{ss}}(\lambda^*) = \dfrac{1 + \sin\beta^*}{2\cos^2\beta^*}\tau
	\end{equation}
	where $ \beta^* \in\left(0,\nicefrac{\pi}{2}\right)$ is the unique solution of $ \beta = \cos\beta $.
\end{prop}
\begin{proof}
	Setting the derivative of~\eqref{eq:variance} equal to zero leads straightly to~\eqref{eq:optGainExplicit}.
	See 
	\isExtended{%
		Appendix B in~\cite{2021arXiv210110394B}
	}{%
	\cref{app:optGainExplicit}
	}%
	for details.
\end{proof}
\autoref{fig:sigma_of_lambda}--\ref{fig:lambda_star_of_tau} show the variance $ \sigma^2_{\textit{ss}}(\lambda) $
and the optimal variance $ \optvar $ and gain $ \opteig $ as functions of the delay, respectively.

\section{Optimization of feedback gains}\label{sec:optimization}

\begin{figure}
	\centering
	\includegraphics[width=0.7\linewidth]{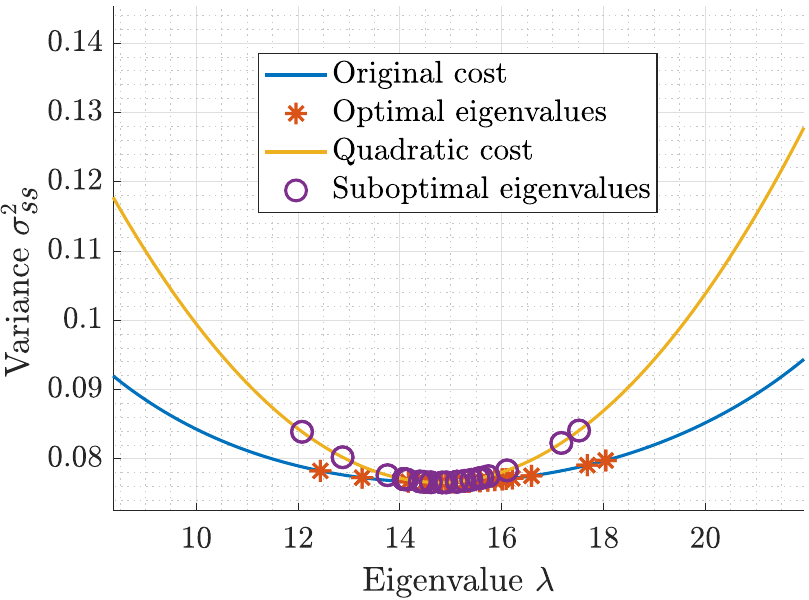}
	\caption{Cost functions in~\eqref{eq:problem-recast} and in~\eqref{eq:problem-recast-eig-squares} 
		about the minimum.}
	\label{fig:comparison_costs}
\end{figure}

In this section, we aim to optimize the feedback gains in~\eqref{eq:feedbackMatrixGains} to minimize the steady-state scalar variance of $ \x{}{t} $. 
\begin{prob}\label{prob}
	Given system~\eqref{eq:stateErrorDynamics} with 
	communication as per~\cref{ass:hypothesis},
	find the feedback gains $ k^\top \doteq [k_1,\dots,k_n] $ 
	minimizing the steady-state scalar variance of $ \x{}{t} $:
	\begin{equation}\label{eq:problem}
		\argmin_{k\in\Real{n}}\; \mathbb{E}\left[\|\xinf\|^2\right]
	\end{equation}
	where $ \xinf \doteq \lim_{t\rightarrow\infty}\x{}{t} $.
\end{prob}

\myParagraph{Problem reformulation}
In virtue of the change of basis presented in~\autoref{sec:eigs-subsystems}, we can write
\begin{equation}\label{eq:varianceChangeOfBasis}
	\mathbb{E}\left[\|\xinf\|^2\right] = \mathbb{E}\left[\|\xinftilde\|^2\right] = \sum_{i=2}^{N} \sigma^2_{\textit{ss}}(\lambda_i)
\end{equation}
where the contribution of $ \lambda_1 = 0 $ is neglected in virtue of the trivial dynamics of $ \xtilde{1}{t} $.
\cref{eq:problem} can be rewritten as
\begin{equation}\label{eq:problem-recast}
	\begin{aligned}
		\argmin_{k\in\Real{n}} &\quad \sum_{i=2}^{N} \sigma^2_{\textit{ss}}(\lambda_i(k))\\
		\mathrm{s.t.} &\quad \lambda_M(k) < \frac{\pi}{2\taun}
	\end{aligned}
\end{equation}
where $ \lambda_M \doteq\max_{i=2}^N\lambda_i $ and the constraint ensures stability
(c.f.~\eqref{eq:retarded-eq-steady-state-condition}).
In virtue of~\eqref{eq:eigenvaluesCirculant} and~\cref{lemma:convexVariance},
problem~\eqref{eq:problem-recast} is convex
and the optimal gains can be found numerically.
To achieve analytical intuition,
we shift to the (sub-optimal) quadratic optimization of the variance arguments:
\begin{equation}\label{eq:problem-recast-eig-squares}
	\begin{aligned}
		\argmin_{k\in\Real{n}}&\quad \sum_{i=2}^N \left(\lambda_i(k)-\opteig\right)^2\\
		\mathrm{s.t.} &\quad \lambda_M(k) < \dfrac{\pi}{2\taun}
	\end{aligned}
\end{equation}
The intuition behind the above reformulation is that the quadratic cost function in~\eqref{eq:problem-recast-eig-squares}
approximates well the variance in~\eqref{eq:problem-recast} about the minimum
(up to scaling and translation).
In particular, the spectrum of $ K $ needs to be ``close" to the optimal gain $ \opteig $
because the variance grows quickly with the smallest and largest eigenvalues (c.f.~\autoref{fig:sigma_of_lambda}).
\autoref{fig:comparison_costs} compares the two cost functions in~\eqref{eq:problem-recast}--\eqref{eq:problem-recast-eig-squares}
for $ N = 50 $, $ n = 10 $ and $ \taun = 0.1 $,
together with their optimal eigenvalues in the single-parameter case (see~\autoref{sec:single-param}).
\begin{rem}[Control regularization]\label{rem:regularization}
	The cost in~\eqref{eq:problem} may be modified to penalize excessive control efforts.
	For example, a weighted norm of $ k $ can be added 
	without affecting convexity.
	This also allows to compare~\eqref{eq:stateErrorDynamics} with delay-free dynamics.
\end{rem}

\subsection{Single parameter}\label{sec:single-param}
We first impose that all feedback gains are equal to the parameter $ \alpha > 0 $,
such that~\eqref{eq:feedbackMatrixMulti} simplifies to
\begin{subequations}\label{eq:feedbackMatrixSingle}
	\begin{equation}
		K = K_f(\alpha) + K_f^\top(\alpha)
	\end{equation}
	\begin{equation}
		K_f(\alpha) = \mbox{circ}
		\left(
		n\alpha, -\alpha, \dots, -\alpha, 0, \dots, 0
		\right)
	\end{equation}
\end{subequations}
and the eigenvectors of $ K $ are proportional to $ \alpha $,
\ie $ \lambda_i = g_i\alpha $, $ i = 1,\dots,N $,
where the coefficients $ g_i $ depend on $ N $ and $ n $ according to~\eqref{eq:eigenvaluesCirculant}.\\
The variance minimization~\eqref{eq:problem} is reduced to
\begin{equation}\label{eq:problem-a-sum-recast}
	\begin{aligned}
		\argmin_{\alpha\in\Realp{}} &\quad \sum_{i=2}^{N} \sigma^2_{\textit{ss}}(\lambda_i(\alpha))\\
		\mathrm{s.t.} &\quad \lambda_M(\alpha) < \frac{\pi}{2\taun}
	\end{aligned}
\end{equation}
and the quadratic approximation~\eqref{eq:problem-recast-eig-squares} becomes
\begin{equation}\label{eq:problem-a-sum-recast-eig-squares}
	\begin{aligned}
		\argmin_{\alpha\in\Realp{}}&\quad \sum_{i=2}^N \left(\lambda_i(\alpha)-\opteig\right)^2\\
		\mathrm{s.t.} &\quad \lambda_M(\alpha) < \dfrac{\pi}{2\taun}
	\end{aligned}
\end{equation}
\begin{thm}\label{thm:AlphaSuboptClosedForm}
	The solution of~\eqref{eq:problem-a-sum-recast-eig-squares} is $ \displaystyle \tilde{\alpha}^* = \dfrac{\opteig}{2n+1} $.
\end{thm}
\begin{proof}
	The result follows by computing the critical points of the cost. 
	See 
	\isExtended{%
		Appendix C in~\cite{2021arXiv210110394B}
	}{%
	\cref{app:alphaSuboptClosedForm}
	}%
	for details.
\end{proof}
In view of~\eqref{eq:optGainExplicit},
\cref{thm:AlphaSuboptClosedForm} implies that the sub-optimal gain (control effort)
decreases as $ \nicefrac{1}{n}\,\cdot\,\nicefrac{1}{f(n)} $,
where the first factor is related to the denser communication topology (more feedback information)
and the second,
which depends on the delay increase rate $ f(n) $,
is a direct consequence of the constraint~\eqref{eq:retarded-eq-steady-state-condition}
that embeds a stability condition.

\begin{figure*}
	\centering
	\begin{minipage}[l]{.33\linewidth}
		\centering
		\includegraphics[width=.8\linewidth]{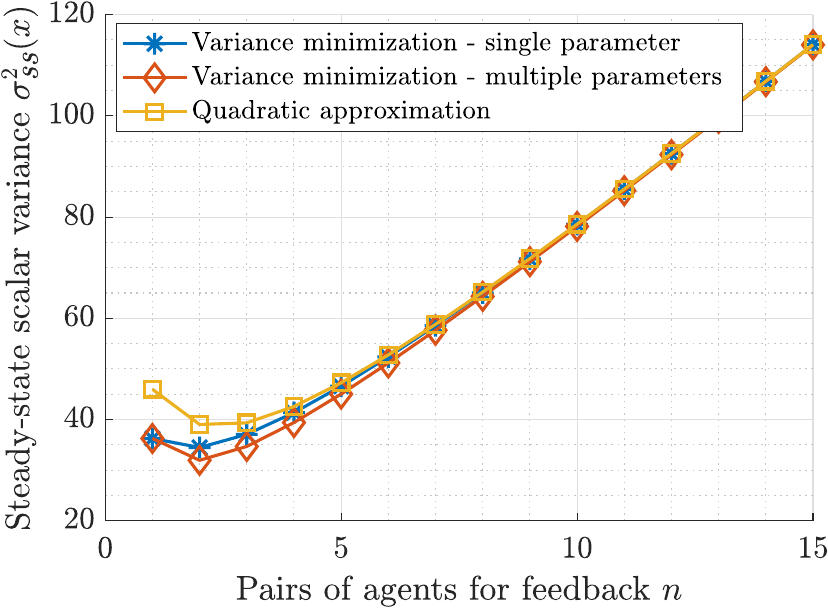}
		\caption{Delay rate $ f(n) = n $.}
		\label{fig:minimization-comp-theo}
	\end{minipage}%
	\hfill
	\begin{minipage}{.33\linewidth}
		\centering
		\includegraphics[width=.8\linewidth]{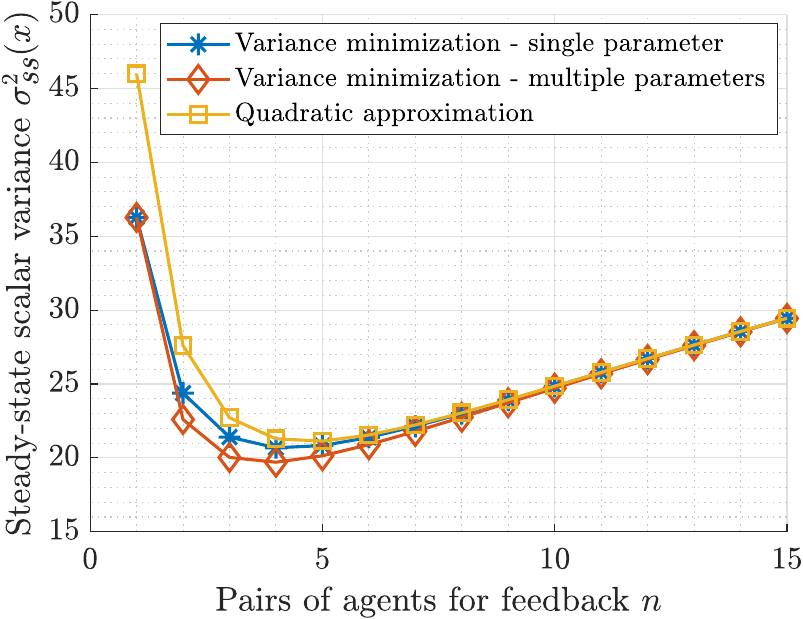}
		\caption{Delay rate $ f(n) = \sqrt{n} $.}
		\label{fig:minimization-comp-sqrt}
	\end{minipage}%
	\hfill
	\begin{minipage}[r]{.33\linewidth}
		\centering
		\includegraphics[width=.8\linewidth]{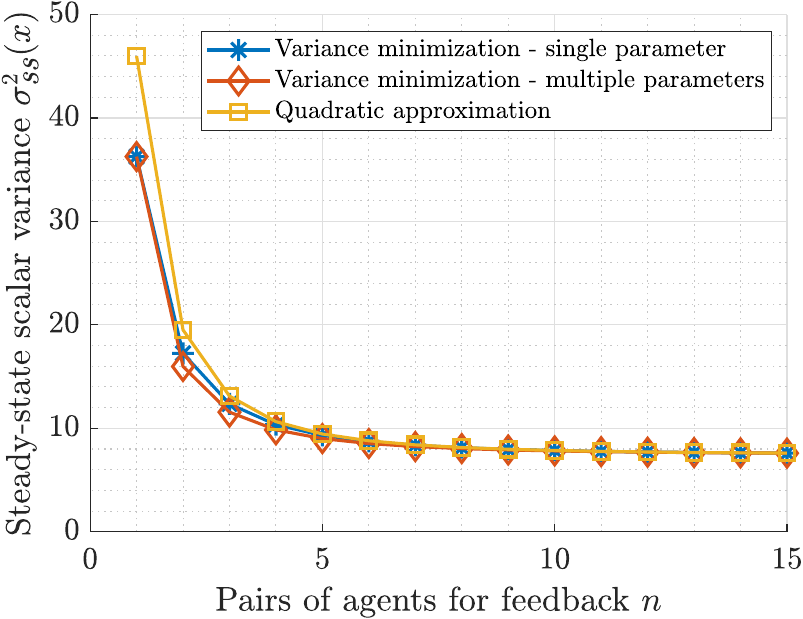}
		\caption{Delay rate $ f(n) \equiv 1 $.}
		\label{fig:minimization-comp-cost}
	\end{minipage}
	\caption{Variance obtained from~\cref{prob} with single (blue) and multiple parameters (red) and
		approximation~\eqref{eq:problem-a-sum-recast-eig-squares} (yellow).}
	\label{fig:minimization-comp}
\end{figure*}

\subsection{multiple parameters}\label{sec:multiple-params}
We now let $ k_j \neq k_i $ such that $ k $ can be any vector in $ \Real{n} $.
\begin{thm}\label{thm:suboptKClosedForm}
	The solution of~\eqref{eq:problem-recast-eig-squares} is $ \tilde{k}^* = \mathds{1}_n\tilde{\alpha}^* $.
\end{thm}
\begin{proof}
	The result follow from properties of the DFT applied to the cost function.
	See 
	\isExtended{%
		Appendix D in~\cite{2021arXiv210110394B}
	}{%
	\cref{app:suboptKClosedForm}
	}%
	for details.
\end{proof}
\cref{thm:suboptKClosedForm} states that optimizations~\eqref{eq:problem-recast-eig-squares} and~\eqref{eq:problem-a-sum-recast-eig-squares} coincide,
\ie choosing the same gain for all agents is a good choice even if different gains are allowed.
~\autoref{fig:minimization-comp} compares the scalar variances obtained
from the two variance minimizations~\eqref{eq:problem-a-sum-recast} and~\eqref{eq:problem}
and from the quadratic approximation~\eqref{eq:problem-a-sum-recast-eig-squares}
with $ N = 50 $, $ \tau_{\textit{min}} = 0.01 $ and three rates $ f(n) $.
\begin{figure}
	\centering
	\includegraphics[width=.7\linewidth]{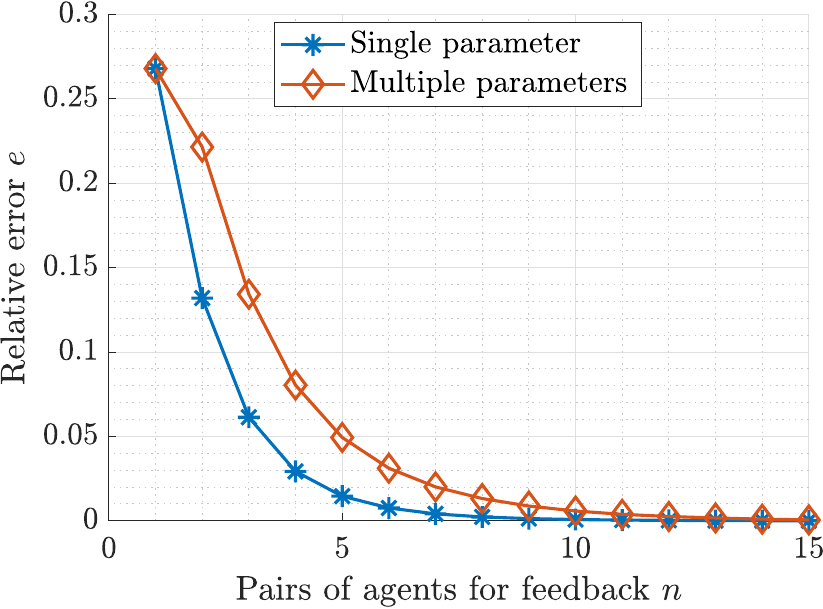}
	\caption{Relative error due to the quadratic approximation~\eqref{eq:problem-a-sum-recast-eig-squares}.}
	\label{fig:relative-error}
\end{figure}
~\autoref{fig:relative-error} shows the relative error (gap with minimum variance) due to the approximation,
defined as
\begin{equation}\label{eq:relativeError}
	e \doteq \dfrac{\optvarx-\sigma_{x,\textit{ss}}^{2,*}}{\sigma_{x,\textit{ss}}^{2,*}}
\end{equation}
where $ \sigma_{x,\textit{ss}}^{2,*} $ is the optimal value of~\eqref{eq:problem-a-sum-recast} or~\eqref{eq:problem}
and $ \optvarx $ is the variance obtained with the quadratic-optimal gain $ \tilde{\alpha}^* $.

\section{discussion: the role of latency in decentralized control}\label{sec:optimal-n}

The gains obtained by solving~\cref{prob}
or its quadratic approximation~\eqref{eq:problem-a-sum-recast-eig-squares},
and the resulting steady-state variance,
depend on  the pairs of communicating agents $ n $ and the delay $ \taun $.
The following result quantifies such dependence for the
solution of the approximated problem~\eqref{eq:problem-a-sum-recast-eig-squares}.
\begin{prop}\label{lem:optVarianceExplicit}
	The scalar variance $ \optvarx $ 
	can be written as
	\begin{equation}\label{eq:optVarianceExplicitNonlinearDelay}
		\optvarx(n) = \tilde{C}^*(n)f(n)\tau_{\textit{min}}
	\end{equation}
	where $ \tilde{C}^*(n) $ does not depend on $ f(n) $ and can be computed exactly.
	The network topology is optimized by
		\begin{equation}\label{eq:optVarOptimalN}
			n^* = \argmin_{n\in\mathbb{N}}\;\optvarx(n)
		\end{equation}
\end{prop}
\begin{proof}
	The result follows from straightforward manipulations of the analytical expression of $ \optvarx $.
	See
	\isExtended{%
		Appendix E in~\cite{2021arXiv210110394B}
		}{%
		\cref{app:optVarianceExplicit}
	}%
	for details.
\end{proof}
Given how $ \tilde{C}^*(n) $ is computed,
characterizing $ \optvarx(n) $ analytically is hard.
Numerical tests 
show that $ \tilde{C}^*(n) $ is submodular decreasing
and that the steady-state variance has a unique point of minimum when $ f(\cdot) $ is discrete-concave,
as shown in~\autoref{fig:minimization-comp}.
\cref{eq:optVarianceExplicitNonlinearDelay} suggests that the optimal
number of links is smaller than or equal to the maximum,
\ie $ n^* \le \floor{\nicefrac{N}{2}}$, 
with the centralized control (corresponding to the complete graph, $ n = \floor{\nicefrac{N}{2}}$)
performing poorly in general.
Also, $ n^* $ depends on the rate $ f(n) $ in a \textquotedblleft non-increasing" fashion,
namely, slower rates yield larger optima.
For example,
the optimal number of links is greater with $ f(n) = \sqrt{n} $ ($ n^* = 5 $,~\autoref{fig:minimization-comp-sqrt})
than with $ f(n) = n $ ($ n^* = 2 $,~\autoref{fig:minimization-comp-theo}).
Centralized control is optimal
when the delay is constant (\autoref{fig:minimization-comp-cost}),
because in this case adding links does not penalize the dynamics.\footnote{Although the same holds true with $ \taun \equiv 0 $, in this case the computations in Appendix E cannot be applied.}

\begin{figure}
	\centering
	\includegraphics[width=0.6\linewidth]{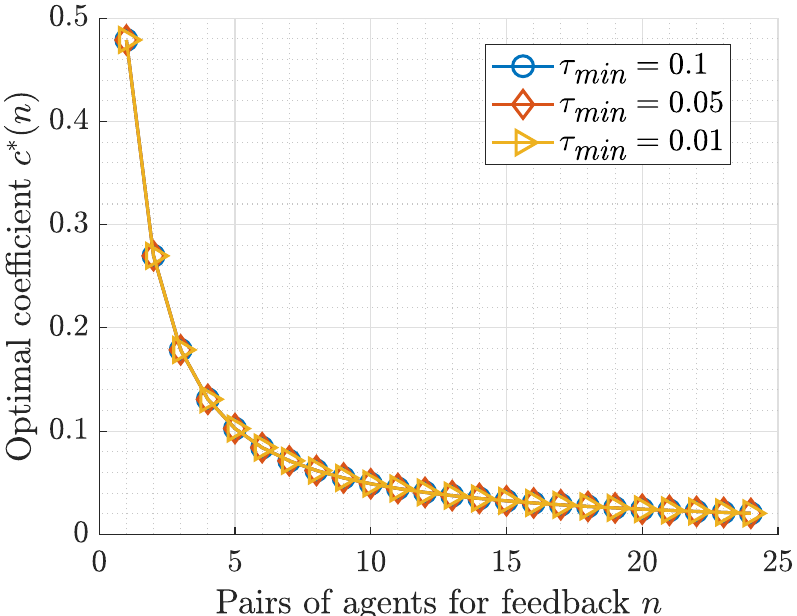}
	\caption{Ratio $ \nicefrac{\alpha^*}{\opteig} $ with different values of $ \tau_{\textit{min}} $.}
	\label{fig:optalphacoeff}
\end{figure}

What discussed so far involves the performance
when the gains are computed according to~\cref{thm:AlphaSuboptClosedForm}.
As shown in the proof of~\cref{lem:optVarianceExplicit},
the key argument to~\eqref{eq:optVarianceExplicitNonlinearDelay} is that
$ \tilde{\alpha}^* $ is linear in $ \opteig $
through a coefficient $ \tilde{c}^*(n) $ that only depends on $ n $. 
Even though proving the same analytically for the optimal gains of~\cref{prob} is hard, 
numerical evidence suggests that this is indeed the case.
For example,~\autoref{fig:optalphacoeff} shows the ratio $ \nicefrac{\alpha^*}{\opteig} $, 
which is independent of both $ \tau_{\textit{min}} $ and $ f(n) $
and hence is consistent with the guess $ \alpha^* = c^*(n)\opteig $.
Simulations show that the relative gap depicted in~\autoref{fig:relative-error}
is also independent of $ f(n) $,
reinforcing such a hypothesis:
in fact, \eqref{eq:relativeError} can be rewritten as 
\begin{equation}\label{eq:relative-error-rewritten}
	e(n) = \dfrac{\tilde{C}^*(n)}{C^*(n)} - 1
\end{equation}
only if the optimal variance can be expressed as $ \sigma_{x,\textit{ss}}^{2,*} = C^*(n)f(n)\tau_{\textit{min}} $.
This suggests that the minimum variance obtained from~\eqref{eq:problem}--\eqref{eq:problem-a-sum-recast}
can also be optimized over $ n $,
as show by~\autoref{fig:minimization-comp},
for both the single-gain and multiple-gain designs.
The optimal network topology 
trades the feedback information (maximized by the centralized architecture)
for the latency (minimized by the fully distributed control):
the optimizer $ n^* $ represents a threshold
beyond which the benefit of adding communication links
is overwhelmed by the delay increase in the retarded dynamics.
\begin{rem}\label{rem:varianceDependsOnN}
	All above results assume a fixed system size $ N $ for ease of notation.
	As one can notice from the proof, the coefficient $ \tilde{C}^*(n) $ is parametric in $ N $,
	and therefore the optimal values $ n^* $ and $ \optvarx(n^*) $ in~\eqref{eq:optVarOptimalN} also depend on $ N $.
\end{rem}
\begin{rem}\label{rem:minVarDoesNotDependOnTauMin}
	\cref{eq:optVarianceExplicitNonlinearDelay} shows that $ n^* $ does not depend on $ \tau_{\textit{min}} $.
	In particular, the optimal topology is only determined by the delay increase rate $ f(n) $
	and by the system size $ N $.
\end{rem}

\section{Conclusions}\label{sec:conclusions}

In this paper, motivated by delay networked control systems,
we design minimum-variance feedback control
assuming that the latency in information exchange increases with the amount of links,
\ie the number of agents involved in the decentralized feedback loops.
Sub-optimal gains,
whose performance tends to the optimum as the network topology becomes denser,
are characterized by a simple expression.
We show that the network topology can be optimized over the number of links,
and that centralized control yields poor performance in general.
Future improvements might involve 
a more complex model for the system dynamics
or heterogeneous agents with different delays.
	
	\isExtended{}{
		\appendix
		\titleformat{\subsection}
		{\normalfont\normalsize\itshape}
		{\thesubsection}
		{1em}
		{}
		
		\titleformat{\section}
		{\normalfont\normalsize\centering\scshape}
		{\thesection}
		{1em}
		{}

\subsection{Proof of~\cref{lemma:convexVariance}}\label{app:convexVariance}

\begin{figure*}
	\centering
	\begin{equation}\label{eq:variance2ndDerivative}
		\dfrac{d^2\sigma_{\textit{ss}}^2}{da^2} =
		\dfrac{(1+2\sin a + a\cos a-\cos(2a))a^2\cos^2a+(a-\cos a+a\sin a-\sin a\cos a)(-2a\cos^2a+2a^2\cos a\sin a)}{a^4\cos^4a}
	\end{equation}
	\noindent\makebox[\linewidth]{\rule{\linewidth}{0.4pt}}
\end{figure*}

Because $ \tau $ does not impact convexity,
we let $ \tau = 1 $.
The second derivative of $ \sigma_{\textit{ss}}^2(a) $, 
written in~\eqref{eq:variance2ndDerivative} in the next page,
is positive if
\begin{multline}\label{eq:positiveDerivative}
	a^3\cos^3a+2a\cos^3a+2a^3\cos a\sin a+2a\cos^3a\sin a+\\
	+2a^3\cos a\sin^2a-2a^2\cos^2a\sin a-2a^2\cos^2a > 0
\end{multline}
The constraint~\eqref{eq:retarded-eq-steady-state-condition} ensures that $ a $, $ \cos a $ and $ \sin a $ are positive.
We now consider three cases for $ a $ and
show that the negative terms in~\eqref{eq:positiveDerivative} are always outbalanced by the positive terms.
\begin{description}[leftmargin=*]
	\item[\boldmath$ a > 1 $:]
		\begin{gather}\label{eq:convexVarianceCase1}
			2a^3\cos a\sin^2a > 2a^2\cos^2a\sin a\\
			2a^3\cos a\sin a > 2a^2\cos^2a
		\end{gather}
	\item[\boldmath$ \cos a < a \le 1 $:]
		\begin{gather}\label{eq:convexVarianceCase2}
			2a^3\cos a\sin a> 2a^2\cos^2a\sin a\\
			2a^3\cos a\sin^2a + 2a\cos^3a \ge 2a^3\cos a > 2a^2\cos^2a
		\end{gather}
	\item[\boldmath$ a \le \cos a $:]
		\begin{gather}\label{eq:convexVarianceCase3}
			2a\cos^3a\sin a \ge 2a^2\cos^2a\sin a\\
			2a\cos a^3 \ge 2a^2\cos^2a
	\end{gather}
\end{description}

\subsection{Proof of~\cref{prop:optimumGain}}\label{app:optGainExplicit}

Uniqueness of $ \opteig $ follows from strict convexity.
The derivative of the variance is
\begin{equation}\label{eq:varianceDerivativeToZero}
	\dfrac{d\sigma_{\textit{ss}}^2}{d\lambda} = 
	\dfrac{\tau\lambda-\cos(\tau\lambda)+\tau\lambda\sin(\tau\lambda)-\cos(\tau\lambda)\sin(\tau\lambda)}{2\lambda^2\cos^2(\tau\lambda)}
\end{equation}
which is equal to zero if and only if
\begin{equation}\label{eq:varianceDerivativeNumToZero}
	(1 + \sin(\tau\lambda))(\tau\lambda - \cos(\tau\lambda)) = 0
\end{equation}
Because of the constraint in~\eqref{eq:opt-var-eig},
\eqref{eq:varianceDerivativeNumToZero} is satisfied if and only if $ \tau\lambda = \cos(\tau\lambda) $.
Applying the change of variable $ \beta \gets \tau\lambda $,
the resulting equation admits a unique solution in $ (0,\nicefrac{\pi}{2}) $.

\subsection{Proof of~\cref{thm:AlphaSuboptClosedForm}}\label{app:alphaSuboptClosedForm}

According to~\eqref{eq:eigenvaluesCirculant}, we first rewrite each eigenvalue as $ \lambda_i = g_i \alpha $.
The cost function in~\eqref{eq:problem-a-sum-recast-eig-squares} can then be rewritten as
\begin{equation}\label{eq:costQuadraticRewritten}
	C(\alpha) = \sum_{i=2}^N \left(g_i\alpha - \opteig\right)^2
\end{equation}
We first exploit strict convexity to find the global minimum of~\eqref{eq:costQuadraticRewritten}.
To this aim, we set
\begin{equation}\label{eq:costAlphaDerivative}
	C'(\alpha) = 2\sum_{i=2}^N g_i\left(g_i\alpha - \opteig\right) = 0
\end{equation}
which admits the unique solution
\begin{equation}\label{eq:optAlpha}
	\tilde{\alpha}^* = \dfrac{\sum_{i=2}^N g_i}{\sum_{i=2}^N g_i^2}\,\opteig
\end{equation}
The coefficients $ g_i $ are the eigenvalues of $ K $ when $ \alpha = 1 $.
Because the eigenvalues of a circulant matrix are the Discrete Fourier Transform (DFT) of its first row,
we can write
\begin{gather}
	\sum_{i=2}^N g_i = \sum_{i=1}^N g_i = N\sum_{i=1}^N r_i = 2Nn \label{eq:sumOfDFT}\\
	\sum_{i=2}^N g_i^2 = \sum_{i=1}^N g_i^2 = N\sum_{i=1}^N r_i^2 = N(4n^2+2n) \label{eq:powerOfDFT}
\end{gather}
where $ r $ is the first row of $ K $,~\eqref{eq:sumOfDFT} comes from the definition of inverse DFT and~\eqref{eq:powerOfDFT} from Plancherel theorem.
The final expression of $ \tilde{\alpha}^* $ follows by substituting~\eqref{eq:sumOfDFT}--\eqref{eq:powerOfDFT} in~\eqref{eq:optAlpha}.\\
We now need to check if such solution satisfies the constraint.
We first note that $ \opteig < \nicefrac{\pi}{4\taun} $
by studying the sign of~\eqref{eq:varianceDerivativeToZero}. 
We then have the following relations for the maximum eigenvalue: 
\begin{align}
	\begin{split}
		\tilde{\lambda}_M^* &= g_M\tilde{\alpha}^* = 2\tilde{\alpha}^*\left(n-\sum_{\ell=1}^n\cos\left(\dfrac{2\pi (M-1)\ell}{N}\right)\right) \\
							&= \dfrac{2\opteig}{2n+1}\left(n-\sum_{\ell=1}^n\cos\left(\dfrac{2\pi (M-1)\ell}{N}\right)\right) \\
							&< \dfrac{\pi}{4\taun}\dfrac{4n}{2n+1} < \dfrac{\pi}{2\taun}
	\end{split}
\end{align}

\subsection{Proof of~\cref{thm:suboptKClosedForm}}\label{app:suboptKClosedForm}

Exploiting linearity of the DFT and Plancherel theorem,
problem~\eqref{eq:problem-recast-eig-squares} can be recast as follows:
\begin{equation}\label{eq:k-sum-simplified}
	\begin{aligned}
		\argmin_{k\in\Real{n}} &\quad \lVert r(k) - \opteig e_1 \rVert^2_2\\
		\mathrm{s.t.} &\quad \lambda_M(k) < \dfrac{\pi}{2\taun}
	\end{aligned}
\end{equation}
where $ r^\top(k) $ is the first row of $ K $
and $ e_\ell $ is the $ \ell $-th canonical vector in $ \Real{N} $.
Given $ \mathcal{I} \subset [N] $, consider now the problem
\begin{equation}\label{eq:probDFT}
	\begin{aligned}
		x^* = \argmin_{x\in\Real{N}} &\quad \lVert x - e_\ell \rVert^2_2\\
		\mathrm{s.t.} 	&\quad x_i = 0 \quad \forall i \in \mathcal{I}\\
						&\quad \sum_{i\notin\mathcal{I}} x_i = 0
	\end{aligned}
\end{equation}
with $\ell\notin\mathcal{I} $ and $ |\mathcal{I}| = N-n-1 $.
We show that \linebreak $ x_i^* \equiv \bar{x}^* \, \forall i \notin \mathcal{I}\cup\{\ell\}$,
which implies that the optimal gains $ k^*_i $ in $ r(k^*) $ are equal.
Assume there exists $ j\notin\mathcal{I}\cup\{\ell\} $ such that
$ x_j \neq x_i \equiv \bar{x} $ for all $ i\notin\mathcal{I}\cup\{j,\ell\} $.
The cost of $ x $ is
\begin{equation}\label{eq:costX}
	C_{x} = \lVert x - e_\ell \rVert^2_2 = \left(x_\ell-1\right)^2 + (n-1)\bar{x}^2 + x_j^2
\end{equation}
Let $ \tilde{x} $ such that $ \tilde{x}_\ell = x_\ell $, $ \tilde{x}_i \equiv \bar{\tilde{x}} $ for all $ i \notin \mathcal{I}\cup\{\ell\} $.
We have
\begin{equation}\label{eq:xTildeEll}
	\begin{aligned}
		x_\ell 			&= -x_j - \sum_{i\neq\ell,j} x_i = -x_j - (n-1)\bar{x}\\
		\tilde{x}_\ell  &= -\sum_{i\neq\ell} \tilde{x}_i = -n\bar{\tilde{x}}
	\end{aligned}
\end{equation}
and the cost associated with $ \tilde{x} $ is
\begin{equation}\label{eq:costXTilde}
	C_{\tilde{x}} = \left(\tilde{x}_\ell-1\right)^2 + n\bar{\tilde{x}}^2 = \left(x_\ell-1\right)^2 + \dfrac{[(n-1)\bar{x}+x_j]^2}{n}
\end{equation}
We then have the following chain of inequalities:
\begin{align}
	\begin{split}
		[(n-1)\bar{x}+x_j]^2 &< n[(n-1)\bar{x}^2+x_j^2]\\
		2(n-1)\bar{x}x_j &< (n-1)\left(\bar{x}^2+x_j^2\right)\\
		0 &< \left(\bar{x}-x_j\right)^2
	\end{split}
\end{align}
which implies $ C_{\tilde{x}} < C_x $ for any $ x,\tilde{x} $.\\
The statement follows from optimality of $ \tilde{\alpha}^* $ in~\eqref{eq:problem-a-sum-recast-eig-squares}.

\subsection{Proof of~\cref{lem:optVarianceExplicit}}\label{app:optVarianceExplicit}

The proof follows straightforward manipulations of the analytical expression of $ \optvarx $.
From~\eqref{eq:varianceChangeOfBasis}, we have
\begin{equation}\label{eq:scalarVarianceSum}
	\optvarx = \sum_{i=2}^N \sigma_{\textit{ss}}^2(\tilde{\lambda}_i^*)
\end{equation}
where $ \tilde{\lambda}_i^* = g_i(n) \tilde{\alpha}^* $ are the near-optimal eigenvalues (c.f. Appendix~\ref{app:alphaSuboptClosedForm}).
Consider the formula in~\cref{thm:AlphaSuboptClosedForm} for the gain $ \tilde{\alpha}^* $,
which we rewrite as $ \tilde{\alpha}^* = \tilde{c}^*(n)\opteig $.
Then, each variance $ \sigma_{\textit{ss}}^2(\tilde{\lambda}_i^*) $ can be written as follows:
\begin{align}\label{eq:optimalVarianceLambda_i}
	\begin{split}
		\sigma_{\textit{ss}}^2(\tilde{\lambda}_i^*) &= \dfrac{1+\sin(\tilde{\lambda}_i^* \taun)}{2\tilde{\lambda}_i^*\cos(\tilde{\lambda}_i^* \taun)} \\
													&= \dfrac{1+\sin(g_i(n)\tilde{c}^*(n)\opteig \taun)}{2g_i(n)\tilde{c}^*(n)\opteig\cos(g_i(n)\tilde{c}^*(n)\opteig \taun)} \\
													&= \dfrac{1+\sin(a_i^*(n)\beta^*)}{2a_i^*(n)\beta^*\cos(a_i^*(n)\beta^*)}\taun = \tilde{C}_i^*(n)\taun
	\end{split}
\end{align}
where~\eqref{eq:optGainExplicit} has been used,
$ a_i^*(n) \doteq g_i(n)\tilde{c}^*(n) $ and
$ \tilde{C}_i^*(n) $ multiplies $ \taun $ in the third line of~\eqref{eq:optimalVarianceLambda_i}.
The scalar variance can then be written as
\begin{equation}\label{eq:scalarVarianceSumExplicit}
	\optvarx = \sum_{i=2}^N \tilde{C}_i^*(n)\taun = \tilde{C}^*(n)f(n)\tau_{\textit{min}}
\end{equation}
where $ \displaystyle \tilde{C}^*(n) \doteq \sum_{i=2}^N \tilde{C}_i^*(n) $.

	}
	
	\isExtended{%
		\bibliographystyle{plain}
		\bibliography{bibfile}
	}{

	}

\end{document}